\documentclass[twocolumn]{svjour3}          
\smartqed  
\usepackage{graphicx}
\usepackage{mathptmx}      
\usepackage{amssymb}
\usepackage{dcolumn}
\usepackage[utf8]{inputenc}
\usepackage{bm}
\newcommand{\be}{\begin{equation}}
\newcommand{\ee}{\end{equation}}
\journalname{Nonlinear Dinamics}

\begin{document}

\title{ Stability of Fixed Points in Generalized Fractional Maps of the Orders $0< \alpha <1$ }


\author{Mark Edelman}

\institute{Mark Edelman \at
              Department of Physics, Stern College at Yeshiva University,   245 Lexington Ave, New York, NY 10016, USA \\
Courant Institute of
Mathematical Sciences, New York University, 251 Mercer St., New York, NY
10012, USA
              \email{edelman@cims.nyu.edu}          
}

\date{Received: date / Accepted: date}

\maketitle

\begin{abstract}
Caputo fractional (with power-law kernels) and fractional (delta) difference maps belong to a more widely defined class of generalized fractional maps, which are discrete convolutions with some power-law-like functions. 
The conditions of the asymptotic stability of the fixed points for maps of the orders $0< \alpha <1$ that are derived in this paper are narrower than the conditions of stability for the discrete convolution equations in general and wider than the well-known conditions of stability for the fractional difference maps. The derived stability conditions for the fractional standard and logistic maps coincide with the results previously observed in numerical simulations. 
In nonlinear maps, one of the derived limits of the fixed-point stability coincides with the fixed-point - asymptotically period two bifurcation point.

\keywords{Fractional maps \ fixed points \ bifurcations}
\PACS{ 05.45.Pq \and 45.10.Hj }
\subclass{47H99 \ 60G99 \and 34A99 \and 39A70 }
\end{abstract}

\section{Introduction}
\label{sec:1}

Fractional maps as the exact solutions of fractional differential equations (of the orders $\alpha>1$) of periodically kicked systems were introduced in \cite{TarZas2008} in a way similar to the way in which the universal map is introduced in regular dynamics. It turned out that these maps are discrete convolutions with power-law functions and are maps with memory whose history is briefly discussed in \cite{TarZas2008} (for more discussions and references to maps with memory see \cite{Chaos2015} and reviews \cite{HBV2,MErev2014}). The results of \cite{TarZas2008} were extended in \cite{MEChaos2013} to include the maps of the orders $ 0 \le \alpha \le 1$. Fractional difference/sum operators were introduced in \cite{MillerRoss,GrayZhang}. The authors of \cite{WuFall} proposed to define fractional maps (which we will call fractional difference maps) as solutions of fractional difference equations (see \cite{AE1,AE2,Anastas,CLZ,Abdel}).

Investigations of the discrete convolution equations, which are analogs of the Volterra integrodifferential equations, were conducted independently from the investigations of fractional maps and were mostly related to the modeling in population dynamics (see, e.g., Chapter 6 in \cite{ElaydiBook}). 

The problem of the linear asymptotic stability of the fixed points of the discrete convolution equations with arbitrary kernels was investigated in 
\cite{Elaydi1993} (for more reviews and publications on the subject see \cite{ElaydiBook,Elaydi1996,Elaydi2007,Elaydi2009,ConvStab2017}). Semi-analytic and numerical results for the case of fractional (power-law kernel) standard and logistic maps were obtained in \cite{ME2,Chaos} (see also reviews \cite{MErev2014,HBV4}).  Later, the problem of the linear asymptotic stability was considered for the case of fractional difference equations (falling factorial kernel)
\cite{AbuSaris2013,Cermak2015,Mozyrska2015} and, recently, for the case of the complex order fractional difference equations \cite{StabComplex}.
In the general case of the discrete convolution equations, the conditions of stability were formulated as the requirements on zeros of a characteristic equation. In the case of the fractional difference equations (maps), the authors of the above cited publications found the explicit analytic conditions of stability. 

Not only the conditions of stability, but also the rates of convergence to the fixed points have been investigated in the discrete convolution equations. It has been proven that if the zero solution of a linear Volterra difference equation is uniformly asymptotically stable, then it is exponentially stable if and only if the kernel decays exponentially (see Theorem 5 in \cite{Elaydi1996}). Semi-analytically and numerically   
the power-law convergence (which for $0<\alpha<1$ is $n^{-\alpha}$, where $n$ is the number of iterations and $\alpha$ is the fractional order of a map) and divergence of trajectories was analyzed in \cite{ME2} (Section 1 and Fig. 1), \cite{ME3} (see Fig. 1), \cite{ME4} (Section 3.2), 
\cite{MErev2014} (Section 3.3.3 and figures in it), and \cite{Chaos2014} (see Fig. 5 for $0< \alpha<1$)  for the fractional and fractional difference standard map. 
For the fractional logistic map, the power-law convergence was shown in 
\cite{MEChaos2013} (see Fig.~6), and for the fractional and fractional difference cases with $0< \alpha<1$, the $n^{-\alpha}$--convergence of trajectories was demonstrated in \cite{ME11} (see Fig. 6.8).
In \cite{Cermak2015} the authors proved that in Caputo fractional difference maps convergence to asymptotically stable fixed points is $O(n^{-\alpha})$ and in \cite{Anh} the authors strictly proved that convergence obeys the power law $\Delta x \sim n^{-\alpha}$, where $0<\alpha<1$ is the order of a fractional difference map. This implies that the correctly calculated Lyapunov exponents in converging to the fixed point Caputo fractional 
difference maps should be equal to zero. However, starting from the first paper where the Lyapunov exponent was used to analyze stability in fractional difference maps \cite{WBLya} (see Figs. 1 and 2), multiple authors reported the negative Lyapunov exponents in various fractional difference maps (see, e.g., Fig.~2 in \cite{DD} or Fig.~7 in \cite{GG}).    

In the following sections, we recall the basic definitions and results related to the problem of stability in the generalized fractional map (Section~\ref{sec:2}), generalize the results known for the problem of stability in fractional difference maps to obtain the criteria of stability of the fixed points in the generalized fractional map (Section~\ref{sec:3}), consider application of these criteria to generalized, fractional, and fractional difference logistic and standard maps (Section~\ref{sec:4}), and make the concluding remarks in  (Section~\ref{sec:5}). 

In this paper we consider only the Caputo fractional and fractional (delta) difference maps and the generalized (as defined in Section~\ref{sec:2}) fractional map. For the sake of brevity, in what follows we will omit the word Caputo.

\section{Preliminaries}
\label{sec:2}

First, we will recall how the generalized fractional map is defined.

When $0<\alpha<1$, the generalized universal $\alpha$-family of maps 
can be written in the form (see \cite{Helman} and Sections~2~and~3 in \cite{ME14}):
\begin{eqnarray}
x_{n}= x_0 
-\sum^{n-1}_{k=0} G^0(x_k) U_\alpha(n-k).
\label{FrUUMapN}
\end{eqnarray} 
In this formula $G^0(x)=h^\alpha G_K(x)/\Gamma(\alpha)$, $x_0$ is the initial condition, $h$ is the time step of the map, $\alpha$ is the order of the map, $G_K(x)$ is a nonlinear function depending on the parameter $K$, $U_\alpha(n)=0$ for $n \le 0$, and $U_\alpha(n) \in \mathbb{D}^0(\mathbb{N}_1)$. The space $\mathbb{D}^i(\mathbb{N}_1)$ is defined as (see \cite{Helman})
{\setlength\arraycolsep{0.5pt}
\begin{eqnarray}
&&\mathbb{D}^i(\mathbb{N}_1)\ \ = \ \ \{f: |\sum^{\infty}_{k=1}\Delta^if(k)|>N, \ \ \forall N, \ \ N \in \mathbb{N}, 
\nonumber \\
&&\sum^{\infty}_{k=1}|\Delta^{i+1}f(k)|=C, \ \ C \in \mathbb{R}_+\},
\label{DefForm}
\end{eqnarray}
}
where $\Delta$ is a forward difference operator defined as
\begin{equation}
\Delta f(n)= f(n+1)-f(n).
\label{Delta}
\end{equation}
The linearized near fixed point $x_f$ map is
\begin{eqnarray}
\delta_{n}= \delta_0 
+\lambda\sum^{n-1}_{k=0} \delta_k U_\alpha(n-k),
\label{Lin0GFMap}
\end{eqnarray}
where  
\begin{eqnarray}
\lambda=-\frac{d G^0(x)}{dx}\Big|_{x=x_f}.
\label{lambda}
\end{eqnarray}
Without loss of generality, we will assume $\delta_0=1$ (equivalent to the following substitution: $\tilde{\delta}=\delta/\delta_0$).
The unilateral $Z$-transform is used to analyze the local stability near fixed points:  
\begin{eqnarray}
Y(z)=\mathbb{Z}[\delta_n](z)= \sum^{\infty}_{k=0} \delta_kz^{-k},
\label{Z}
\end{eqnarray}
where $z\in C$.
$Z$-transformed Eq.~(\ref{Lin0GFMap}) with $\delta_0=1$  is
\begin{eqnarray}
Y(z)=\frac{1}{(1-z^{-1})(1-\lambda\mathbb{Z}[U_\alpha(n)](z))},
\label{Zeq}
\end{eqnarray}
where
\begin{eqnarray}
\mathbb{Z}[U_\alpha(n)](z)= \sum^{\infty}_{k=0}U(k)z^{-k}.
\label{ZU}
\end{eqnarray}

In fractional maps 
\begin{eqnarray}
U_{\alpha}(n)=n^{\alpha-1}, \ \ \ \  U_{\alpha}(1)=1
\label{UnFr}
\end{eqnarray} 
and in fractional difference maps
{\setlength\arraycolsep{0.5pt}
\begin{eqnarray}
&&U_{\alpha}(n)=(n+\alpha-2)^{(\alpha-1)} 
, \ \ \ \  \nonumber \\  
&&U_{\alpha}(1)=(\alpha-1)^{(\alpha-1)}=\Gamma(\alpha).
\label{UnFrDif}
\end{eqnarray} 
}
The definition of the falling factorial $t^{(\alpha)}$ is
\begin{equation}
t^{(\alpha)} =\frac{\Gamma(t+1)}{\Gamma(t+1-\alpha)}, \ \ t\ne -1, -2, -3.
...
\label{FrFacN}
\end{equation}
The falling factorial is asymptotically a power function:
\begin{equation}
\lim_{t \rightarrow
  \infty}\frac{\Gamma(t+1)}{\Gamma(t+1-\alpha)t^{\alpha}}=1,  
\ \ \ \alpha \in  \mathbb{R}.
\label{GammaLimitN}
\end{equation}
The $h$-falling factorial $t^{(\alpha)}_h$ is defined as
\begin{eqnarray}
&&t^{(\alpha)}_h =h^{\alpha}\frac{\Gamma(\frac{t}{h}+1)}{\Gamma(\frac{t}{h}+1-\alpha)}= h^{\alpha}\Bigl(\frac{t}{h}\Bigr)^{(\alpha)},
\nonumber \\
&&\frac{t}{h} \ne -1, -2, -3,
....
\label{hFrFacN}
\end{eqnarray}
Following \cite{Cermak2015,Mozyrska2015}, we notice that in the fractional difference case 
\begin{equation}
\frac{U_\alpha(n)}{\Gamma(\alpha)}=
\left( 
\begin{array}{c}
n+\alpha-2 \\ n-1
\end{array} \right)=\tilde{\phi}_\alpha(n-1),
\label{FrDifUbc}
\end{equation}
where
\begin{equation}
\tilde{\phi}_\alpha(n)=
\left( 
\begin{array}{c}
n+\alpha-1 \\ n
\end{array} \right)=
(-1)^n
\left( 
\begin{array}{c}
-\alpha \\ n
\end{array} \right)
\label{fi_alpha},
\end{equation}
\begin{equation}
\mathbb{Z}[\tilde{\phi}_\alpha(n)](z)=
\Bigl(\frac{z}{z-1}\Bigr)^\alpha,
\label{phi_alpha}
\end{equation}
and 
\begin{equation}
\mathbb{Z}[U_\alpha(n)](z)=
\frac{\Gamma(\alpha)}{z}\Bigl(\frac{z}{z-1}\Bigr)^\alpha.
\label{FrDifZU}
\end{equation}

\section{Linear stability of the generalized fractional map of the order $0<\alpha<1$}
\label{sec:3}

If we assume that, as in fractional and fractional difference maps, $U_\alpha(n)>0$ for $n \in \mathbb{N}_1$, then the following generalization of Lemma~3.1, proven in \cite{AbuSaris2013} for the fractional difference maps, is true: 
\begin{lemma}\label{L1} 
If $U_\alpha(n)>0$ for $n \in \mathbb{N}_1$, then a fixed point $x_f$ of the map Eq.~(\ref{FrUUMapN}) is linearly unstable when $\frac{dG^0(x)}{dx}\Big|_{x=x_f}<0$.
\end{lemma}
\begin{proof}
Assuming that $\delta_0=1$, from Eqs.~(\ref{Lin0GFMap})~and~(\ref{lambda}) follows that $\delta_k>1$ and
\begin{eqnarray}
&&\delta_{n}= 1 
-\frac{d G^0(x)}{dx}\Big|_{x=x_f}\sum^{n-1}_{k=0} \delta_k U_\alpha(n-k)>
\nonumber \\
&&1-\frac{dG^0(x)}{dx}\Big|_{x=x_f}\sum^{n}_{k=1}U_\alpha(k).
\label{Lin0GFMapProof}
\end{eqnarray}
Because $U_\alpha(n) \in \mathbb{D}^0(\mathbb{N}_1)$, the last sum tends to infinity when $n \rightarrow \infty$. This ends the proof.
\end{proof}
Application of Lemma~\ref{L1} to the generalized fractional logistic map ($G_K(x)=x-Kx(1-x)$) gives the result known for the fractional and fractional difference logistic maps that the fixed point $x=0$ is unstable when $K>1$ and the fixed point $x=(K-1)/K$ is unstable when $K<1$. For the fixed points $x=\pi m$, $m \in \mathbb{Z}$ of the generalized fractional standard map ($G_K(x)=K\sin(x)$) the result of the Lemma is instability when $(-1)^mK<0$.  

To analyze the stability problem for a system described by Eq.~(\ref{FrUUMapN}), we rewrite Eq.~(\ref{Lin0GFMap}) as 
\begin{eqnarray}
\delta_{n+1}= \delta_{n} 
+\lambda\sum^{n}_{k=0} \delta_k \tilde{U}^1_\alpha(n-k),
\label{Lin0GFMapN}
\end{eqnarray}
where 
{\setlength\arraycolsep{0.5pt} 
\begin{eqnarray}
&&U^{m-1}_\alpha(n)= \Delta^{m-1}U^0_\alpha(n-m+1), \  \  U_\alpha(n)=U^0_\alpha(n) \  \ {\rm and}
\nonumber \\ 
&&\tilde{U}^1_\alpha(k)= U^1_\alpha(k+1).
\label{UmT}
\end{eqnarray}
}
Z-transform of Eq.~(\ref{Lin0GFMapN}) is  
\begin{eqnarray}
&&Y(z)=\frac{z}{z-1-\lambda \mathbb{Z}[\tilde{U}^1_\alpha(n)](z)} 
=\frac{z}{z-1-\lambda z\mathbb{Z}[U^1_\alpha(n)](z)}
\nonumber \\ 
&&=\frac{z}{z-1-(1-z^{-1})\lambda z\mathbb{Z}[U_\alpha(n)](z)}.
\label{ZeqN}
\end{eqnarray}
Because $U_\alpha(n) \in \mathbb{D}^0(\mathbb{N}_1)$, $U^1_\alpha(n) \in l_1$ and its norm is defined as $|| U^1_\alpha(n) ||
=\sum^{\infty}_{k=1}|U^1_\alpha(k)|$. Then, according to Theorem~6.14 of \cite{ElaydiBook}, $\mathbb{Z}[U^1_\alpha(n)](z)$ is analytic and $|\mathbb{Z}[U^1_\alpha(n)](z)| \ge || U^1_\alpha(n) || $ for $|z| \ge 1$. According to Theorem 6.17 of \cite{ElaydiBook}, the zero solution of the system defined by  Eq.~(\ref{Lin0GFMapN}) (and, correspondingly, of the linearized original Eq.~(\ref{FrUUMapN})) is uniformly asymptotically stable if and only if $\omega(z)=z-1-\lambda z \mathbb{Z}[U^1_\alpha(n)](z)$ 
has no zeros for all $|z| \ge 1$. z is a zero of $\omega(z)$ if and only if
\begin{eqnarray}
&&\lambda=\frac{z-1}{z}\frac{1}{\mathbb{Z}[U^1_\alpha(n)](z)}
\nonumber \\ 
&&=\frac{z-1}{z}\frac{1}{\sum^{\infty}_{k=1}[U_\alpha(k)-U_\alpha(k-1)]z^{-k}}
\nonumber \\ 
&&=\frac{z-1}{\sum^{\infty}_{k=0}[U_\alpha(k+1)-U_\alpha(k)]z^{-k}}.
\label{ZeqNlambda}
\end{eqnarray}
The span of the real values of $\lambda$ for which the zero solution is stable corresponds to the values of Eq.~(\ref{ZeqNlambda}) from $z=-1$ to $z=1$. $\lambda=0$ when $z=1$ and when $z=-1$
\begin{eqnarray}
&&\lambda(-1)=\frac{-2}{U_\alpha(1)+ \sum^{\infty}_{k=1}[U_\alpha(k+1)-U_\alpha(k)](-1)^{k}}
\nonumber \\ 
&&=\frac{1}{\sum^{\infty}_{k=1}U_\alpha(k)(-1)^{k}}=-\frac{1}{S_2},
\label{ZeqNlambdaNeg1}
\end{eqnarray}
where we used the definition $S_2=\sum^{\infty}_{k=1}U_\alpha(k)(-1)^{k+1}$ introduced in \cite{ME14}. 
Finally, we may formulate the following condition of stability of the fixed point for the linearized map Eq.~(\ref{Lin0GFMap}):
\begin{eqnarray}
&&-\frac{1}{S_2}<\lambda<0,
\label{ZeqNlambdaStability}
\end{eqnarray}
or, equivalently,
\begin{eqnarray}
&&0<\frac{d G(x)}{dx}\Big|_{x=x_f}<\frac{\Gamma(\alpha)}{S_2 h^\alpha}.
\label{ZeqNlambdaStabilityN}
\end{eqnarray}
This result may be formulated as the following theorem:
\begin{theorem}\label{T1} 
The map Eq.~(\ref{FrUUMapN}), where $G^0(x)=h^\alpha G_K(x)/\Gamma(\alpha)$, $x_0$ is the initial condition, $h$ is the time step of the map, $\alpha$ is the order of the map, $G_K(x)$ is a nonlinear function depending on the parameter $K$, $U_\alpha(n)=0$ for $n \le 0$, and $\Delta U_\alpha(n) \in l_1$ is asymptotically stable if and only if the conditions Eq.~(\ref{ZeqNlambdaStabilityN}), where $S_2=\sum^{\infty}_{k=1}U_\alpha(k)(-1)^{k+1}$,  are satisfied.
\end{theorem}

Let us recall the equations obtained in \cite{ME14} for the calculation of the period two ($T=2$) points $x_1$ and $x_2$:
\begin{eqnarray}
&&x_2-x_1=\frac{S_2h^\alpha}{\Gamma(\alpha)}[G(x_2)-G(x_1)],
\nonumber \\ 
&&G(x_1)+G(x_2)=0.
\label{BifConditionO}
\end{eqnarray}
Near the first bifurcation point $x_f$, where $x_1\approx x_2$, the first of these equations can be written as
\begin{eqnarray}
&&x_2-x_1=\frac{S_2h^\alpha}{\Gamma(\alpha)} \frac{d G(x)}{dx}\Big|_{x=x_f} (x_2-x_1).
\label{BifCondition}
\end{eqnarray}
This equation gives the following condition for the first bifurcation 
\begin{eqnarray}
&& \frac{d G(x)}{dx}\Big|_{x=x_f} = \frac{\Gamma(\alpha)}{S_2h^\alpha}
\label{BifConditionN}
\end{eqnarray}
consistent with Eq.~(\ref{ZeqNlambdaStabilityN}).

In the case of the fractional difference maps, $S_2$ can be calculated explicitly:
{\setlength\arraycolsep{0.5pt}   
\begin{eqnarray} 
&&S_{2}=\sum^{\infty}_{k=0}\Bigl[
U_{\alpha} (2k+1) - U_{\alpha} (2k+2)\Bigr]
\nonumber \\
&&= \sum^{\infty}_{k=0}\Bigl[(2k+\alpha-1)^{(\alpha-1)}- (2k+\alpha)^{(\alpha-1)} \Bigr]
\nonumber \\
&&= \sum^{\infty}_{k=0}
\Bigl[\frac{\Gamma(2k+\alpha)}{ \Gamma(2k+1)}- \frac{\Gamma(2k+\alpha+1)}{ \Gamma(2k+2)} \Bigr]
\nonumber \\
&&=(1-\alpha)\sum^{\infty}_{k=0}
\frac{\Gamma(2k+\alpha)}{\Gamma(2k+2)}
=-\Gamma(\alpha) \sum^{\infty}_{k=0}
\frac{\Gamma(2k+\alpha)}{\Gamma(\alpha-1)\Gamma(2k+2)}
\nonumber \\
&&=-\Gamma(\alpha) \sum^{\infty}_{k=0}
\left( \begin{array}{c}
2k+\alpha-1 \\ 2k+1
\end{array} \right)
= \Gamma(\alpha) \sum^{\infty}_{k=0}
\left( \begin{array}{c}
1-\alpha \\ 2k+1
\end{array} \right)
\nonumber \\
&&=\Gamma(\alpha) \sum^{\infty}_{k=0}
\Bigl[
\left( 
\begin{array}{c}
-\alpha \\ 2k+1
\end{array} \right)
+
\left( 
\begin{array}{c}
-\alpha \\ 2k
\end{array} \right)
\Bigr]
\nonumber \\
&&=\Gamma(\alpha) \sum^{\infty}_{k=0}
\left( 
\begin{array}{c}
-\alpha \\ k
\end{array} \right) =\Gamma(\alpha) 2^{-\alpha}.
\label{S2}
\end{eqnarray}
}
Here we used the following well-known and easily verifiable identities:
{\setlength\arraycolsep{0.5pt}   
\begin{eqnarray} 
&&\left( 
\begin{array}{c}
\mu \\ \eta
\end{array} \right)
=\frac{\Gamma(\mu+1)}{\Gamma(\mu-\eta+1)\Gamma(\eta+1)},
\label{I1}
\end{eqnarray}
}
{\setlength\arraycolsep{0.5pt}   
\begin{eqnarray} 
&&\left( 
\begin{array}{c}
n+\alpha-1 \\ n
\end{array} \right)
= (-1)^n
\left( 
\begin{array}{c}
-\alpha \\ n
\end{array} \right),
\label{I2}
\end{eqnarray}
} 
{\setlength\arraycolsep{0.5pt}   
\begin{eqnarray} 
&&\left( 
\begin{array}{c}
\mu \\ n
\end{array} \right)+
\left( 
\begin{array}{c}
\mu \\ n+1
\end{array} \right)
= 
\left( 
\begin{array}{c}
\mu+1 \\ n+1
\end{array} \right),
\label{I3}
\end{eqnarray}
}
and
{\setlength\arraycolsep{0.5pt}   
\begin{eqnarray} 
&&(x+y)^{\alpha}
= \sum^{\infty}_{k=0}
\left( 
\begin{array}{c}
\alpha \\ k
\end{array} \right) x^{\alpha-k}y^k.
\label{I4}
\end{eqnarray}
} 
Then, in fractional difference maps the stability of the fixed points
is defined by the inequality
\begin{eqnarray}
&&0<\frac{d G(x)}{dx}\Big|_{x=x_f}<\Bigl(\frac{2}{h}\Bigr)^\alpha
\label{ZeqNlambdaStabilityNFD}
\end{eqnarray}
and the bifurcation point is defined by the equality
\begin{eqnarray}
&&\frac{d G(x)}{dx}\Big|_{x=x_f}=\Bigl(\frac{2}{h}\Bigr)^\alpha
\label{BifPointNFD}
\end{eqnarray}


\section{Generalized fractional logistic and standard maps, $0<\alpha<1$}
\label{sec:4}

Fractional logistic (quadratic nonlinearity) and standard (harmonic nonlinearity) maps were among the first introduced fractional and fractional difference maps \cite{MEChaos2013,WuFall,ME2,Chaos2014,ME9,WBLog}. They also were used in applications related to communications \cite{Com1,Com2}, aging \cite{ME11}, and encryption \cite{Encr}.

\subsection{Generalized fractional logistic map, $0<\alpha<1$}
\label{sec:4.1}

In the generalized fractional logistic map 
\begin{eqnarray}
G_K(x)=x-Kx(1-x).
\label{LM}
\end{eqnarray} 
This map has two fixed points $x_{f1}=0$ and $x_{f2}=(K-1)/K$.

The linearized version of this map around the fixed point at the origin ($x=0$) is 
\begin{eqnarray}
\delta_{n}= \delta_0 
+\frac{h^\alpha}{\Gamma(\alpha)}(K-1)\sum^{n-1}_{k=0} \delta_k U_\alpha(n-k).
\label{Lin0GFLMap}
\end{eqnarray} 
The linearized version of this map around the fixed point $x=(K-1)/K$ is 
\begin{eqnarray}
\delta_{n}= \delta_0 
-\frac{h^\alpha}{\Gamma(\alpha)}(K-1)\sum^{n-1}_{k=0} \delta_k    U_\alpha (n-k).
\label{Lin0GFLMapN}
\end{eqnarray}
Eq.~(\ref{ZeqNlambdaStabilityN}) applied to the fixed point $x_{f1}$ gives the following condition of stability:
\begin{eqnarray}
&&1-\frac{\Gamma(\alpha)}{S_2 h^\alpha}<K<1
\label{LogStab0}
\end{eqnarray}
and the condition of stability of the fixed point $x_{f2}$ is
\begin{eqnarray}
&&1<K<\frac{\Gamma(\alpha)}{S_2 h^\alpha}+1.
\label{LogStabX2}
\end{eqnarray}
The bifurcation point is 
\begin{eqnarray}
&&K=\frac{\Gamma(\alpha)}{S_2 h^\alpha}+1.
\label{LogFirstBif}
\end{eqnarray}
In the case the fractional difference logistic map $x_{f1}$ is stable when
\begin{eqnarray}
&&1-\Bigl(\frac{2}{h}\Bigr)^\alpha <K<1,
\label{LogStab0N}
\end{eqnarray}
$x_{f2}$ is stable when
\begin{eqnarray}
&&1<K<\Bigl(\frac{2}{h}\Bigr)^\alpha +1,
\label{LogStabX2N}
\end{eqnarray}
and the bifurcation point is 
\begin{eqnarray}
&&K=\Bigl(\frac{2}{h}\Bigr)^\alpha+1.
\label{LogFirstBifn}
\end{eqnarray}
\begin{figure}[!t]
\includegraphics[width=0.45 \textwidth]{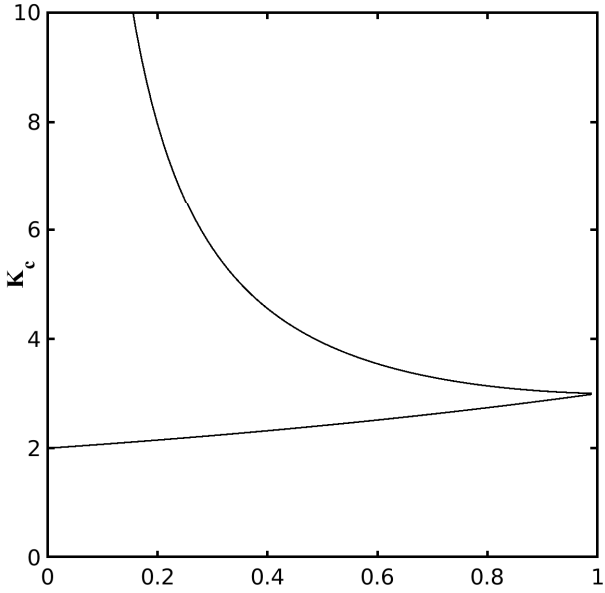}
\vspace{-0.25cm}
\caption{The fixed-point -- T=2-point bifurcations for fractional (upper curve) and fractional difference (lower curve) Caputo logistic maps. }
\label{fig1}
\end{figure}

\subsection{Generalized fractional standard (circle with zero driving phase) map, $0<\alpha<1$}
\label{sec:4.2}

In the generalized fractional standard map 
\begin{eqnarray}
G_K(x)=K\sin(x).
\label{SM}
\end{eqnarray} 
The fixed points are $x_{fn}=\pi n$ and
\begin{eqnarray}
&&\frac{d G(x)}{dx}\Big|_{x=x_{fn}}=(-1)^nK.
\label{SMfixDer}
\end{eqnarray}
The condition of stability is 
\begin{eqnarray}
&&0<(-1)^nK<\frac{\Gamma(\alpha)}{S_2 h^\alpha}.
\label{SMStabX2}
\end{eqnarray} 
The bifurcation points are  
\begin{eqnarray}
&&K=\pm\frac{\Gamma(\alpha)}{S_2 h^\alpha}.
\label{SMBifX2}
\end{eqnarray} 
\begin{figure}[!t]
\includegraphics[width=0.45 \textwidth]{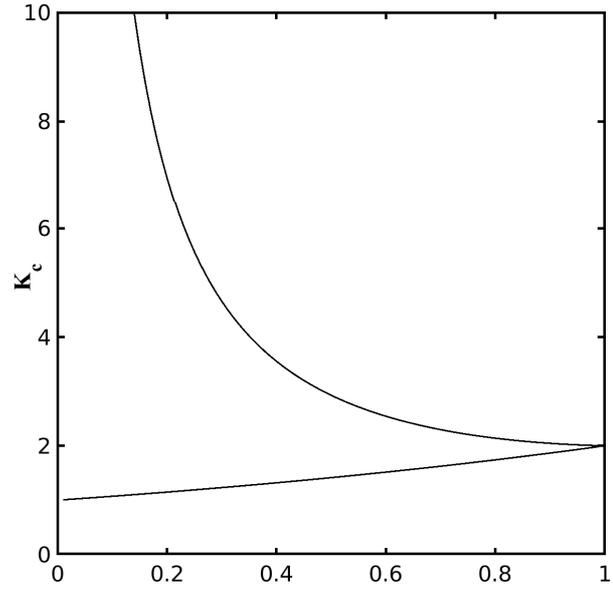}
\vspace{-0.25cm}
\caption{ The fixed-point -- T=2-point bifurcations for fractional (upper curve) and fractional difference (lower curve) Caputo standard maps.}
\label{fig2}
\end{figure}

In the case of the fractional difference standard map
the condition of stability is 
\begin{eqnarray}
&&0<(-1)^nK<\Bigl(\frac{2}{h}\Bigr)^\alpha.
\label{SMStabX2n}
\end{eqnarray} 
The bifurcation points are  
\begin{eqnarray}
&&K=\pm\Bigl(\frac{2}{h}\Bigr)^\alpha.
\label{SMBifX2n}
\end{eqnarray} 

The graphs of the fixed-point -- T=2-point bifurcation in the case of the unit time step $h=1$ Figures~\ref{fig1}~and~2 coincide with the previously reported (see, e.g., \cite{HBV4}) results. In the case of the standard maps, we considered the zero fixed-point. The bifurcation points in the fractional difference maps (lower curves) were calculated using Eqs.~(\ref{LogFirstBifn}) and~(\ref{SMBifX2n}) . The bifurcation points in the fractional maps (upper curves) were calculated using Eqs.~(\ref{LogFirstBif})~and~(\ref{SMBifX2}). The values of $S_2$ were calculated using the algorithm described in \cite{ME14} which is based on the calculation of the Riemann $\zeta$-function.

\section{Conclusion}
\label{sec:5}

The inequality derived in this paper, Eq.~(\ref{ZeqNlambdaStabilityN}), defines the necessary and sufficient conditions for the local stability of fixed points of the nonlinear generalized fractional maps and the upper limit in this inequality defines the fixed-point -- T=2 bifurcation point. The only condition used to derive Eq.~(\ref{ZeqNlambdaStabilityN}) was the following requirement on the kernel $U_\alpha (n)$ of the map 
\begin{eqnarray}
 \Delta U_\alpha (n) \in l_1.
\label{Inl1}
\end{eqnarray} 
This implies that any map, Eq.~(\ref{FrUUMapN}), with the kernel satisfying Eq.~(\ref{Inl1}) will have the same conditions of stability of its fixed points.


The requirement that $\sum^{\infty}_{k=1}U_\alpha(k)=\pm \infty$, which was important for the finding of the periodic points in \cite{Helman,ME14}, was not used in this paper.

\begin{acknowledgements}
The author acknowledges support from Yeshiva University's 2021--2022
Faculty Research Fund and expresses his gratitude to the administration of Courant Institute of Mathematical
Sciences at NYU for the opportunity to perform computations at Courant and to Virginia Donnelly for technical help.
\end{acknowledgements}

\section*{Conflict of interest}
The author declares that he has no conflict of interest.


\end{document}